\documentclass[11pt]{article}

\usepackage[margin=1in]{geometry}
\usepackage{amsmath,amssymb,amsthm}
\usepackage{graphicx}
\usepackage{subcaption}
\usepackage[round,authoryear]{natbib}
\usepackage{newtxtext}
\usepackage[subscriptcorrection]{newtxmath}
\usepackage[colorlinks=true,linkcolor=blue,citecolor=blue,urlcolor=blue]{hyperref}

\newtheorem{theorem}{Theorem}
\newtheorem{corollary}{Corollary}
\theoremstyle{definition}
\newtheorem{definition}{Definition}
\theoremstyle{remark}
\newtheorem{remark}{Remark}

\newcommand{\map}{\varphi}
\newcommand{\cutmap}{\psi}
\newcommand{\dG}{d_G}
\newcommand{\dVI}{d_G^{\mathrm{VI}}}
\newcommand{\dVD}{d_G^{\mathrm{VD}}}
\newcommand{\dMK}{d_G^{(\kappa)}}

\newcommand{\HH}{H}
\newcommand{\II}{I}

\hypersetup{
  pdftitle={On Graph-Informed Distance Metrics for Comparing Graph Partitions},
  pdfauthor={Srijato Bhattacharyya, Huiyan Sang, and Bani Mallick},
  pdfkeywords={community detection, graph partition, Mirkin distance, stochastic block model, van Dongen metric, variation of information}
}

\title{On Graph-Informed Distance Metrics for Comparing Graph Partitions}
\author{Srijato Bhattacharyya, Huiyan Sang, and Bani Mallick\\[0.4em]
\small Department of Statistics, Texas A\&M University\\
\small 3143 TAMU, 155 Ireland Street, College Station, Texas 77843, U.S.A.\\
\small \texttt{srijato@tamu.edu, huiyan@stat.tamu.edu, bmallick@stat.tamu.edu}}
\date{}

\begin{document}
\maketitle

\begin{abstract}
Comparing graph partitions is fundamental to the analysis of network-structured data, yet existing measures for comparing graph partitions typically rely on graph-agnostic indices that treat vertices as exchangeable, ignoring the underlying graph topology that encodes essential information about community cohesion and separation. We propose a general construction of graph-informed distances that compares vertex partitions through induced edge partitions and yields valid metrics on the space of contiguous graph partitions. As special cases, we develop graph-informed versions of variation of information and the van Dongen distance together with a binary cut-based companion distance, and show that these distances satisfy a natural local graph-aware refinement criterion. Under stochastic block models, we prove that stronger topological disruptions incur asymptotically larger distances almost surely in both inter-community and intra-community split settings. These results provide a simple and principled framework to compare graph partitions while respecting the underlying graph structure.
\end{abstract}

\noindent\textbf{Keywords:} community detection; graph partition; Mirkin distance; stochastic block model; van Dongen metric; variation of information

\section{Introduction} 
The clustering of data with graph relations has numerous applications in network community detection \citep{geng2019probabilistic}, spatial clustering \citep{luo2021bayesian}, disease mapping \citep{datta2019spatial}, and image segmentation \citep{felzenszwalb2004efficient}. A distance metric between graph clusterings is a fundamental component for evaluating different graph clustering methods~\citep{meilua2007comparing}, deriving posterior summaries of Bayesian graph clustering results~\citep{wade2018bayesian}, conducting geometric analysis on the metric manifold of graph clustering~\citep{song2026non}, and constructing dependent models for multivariate graph clustering~\citep{paul2020random}.   

For graph clusterings \citep{lancichinetti2009community, bhattacharyya2025constrained}, standard comparisons have largely relied on measures developed for general set partitions. Widely used examples include the Rand index \citep{rand1971objective}, the adjusted Rand index \citep{hubert1985comparing}, the Jaccard coefficient, and the Fowlkes--Mallows index \citep{fowlkes1983method}. These criteria compare node assignments and are useful for measuring agreement in labels, but they are graph-agnostic: they treat vertices as exchangeable and ignore the topological role of edges.
A second line of work compares partitions through metric structure. \citet{dongen2000performance} introduced a bounded metric on the space of set partitions, and this distance has recently been extended and used as a loss over contiguous spatial partitions in Bayesian domain-partitioning problems \citep{huang2025consistent}. Information-theoretic comparisons are also prominent. \citet{meila2005comparing,meilua2007comparing,wade2018bayesian} developed the variation of information as a true metric on partitions, while \citet{vinh2009information} and \citet{romano2014standardized} studied related mutual-information-based normalizations and chance corrections. These approaches remain graph-agnostic when applied directly to graph clusterings. \citet{poulin2020comparing} proposed graph-aware similarities and dissimilarities based on the induced binary classification of edges into within-cluster and between-cluster types, and showed that graph-aware and graph-agnostic criteria can behave very differently under refinement and coarsening.
\citet{yan2025spatially} developed a spatially aware adjusted Rand index. While useful for graph-aware clustering comparison, these two methods do not provide a general construction of metrics on graph partitions or establish theoretical properties characterizing how these metrics respond to topological perturbations. 

Our contribution is to develop an edge-partition construction of graph-informed distances for graph partitions. The construction yields genuine metrics on the space of contiguous graph partitions and naturally accommodates graph-informed versions of variation of information and the van Dongen distance, together with a binary cut-based companion distance. We then formulate a local graph-aware refinement criterion for one-step refinements and show that these three distances satisfy it through distinct local scores. Finally, under stochastic block models (SBM), we establish two topology-aware properties, inter-community split awareness and intra-community split awareness, under which stronger structural disruptions incur larger distances almost surely.

\section{Construction of graph-informed metrics}\label{sec:construction}
For any finite non-empty set $X$, let $\Pi(X)$ denote the collection of all partitions of $X$, that is, the family of finite sets of non-empty, pairwise disjoint subsets whose union is $X$. Let $\phi=\phi_X:\Pi(X)\times\Pi(X)\to\mathbb{R}_{\ge 0}$ denote a distance metric on $\Pi(X)$.

Let $G=(V,E)$ be a simple undirected graph with $n=|V|$ vertices and $m=|E|>0$ edges. For a vertex partition $P=\{P_1,\ldots,P_K\}\in\Pi(V)$, let $E_i(P)=\{\{u,v\}\in E:u,v\in P_i\}$ for $i\geq 1$, and let $E_0(P)=\{\{u,v\}\in E:\text{$u$ and $v$ lie in different clusters of $P$}\}$. The sets $E_i(P)$ may be empty, in particular when $P_i$ is a singleton. For a graph $G$, let $G[S]$ denote the induced subgraph with vertex set $S\subseteq V$ and edge set $\{\{u,v\}\in E:u,v\in S\}$. A partition $P=\{P_1,\ldots,P_K\}\in\Pi(V)$ is called \emph{contiguous} if each $G[P_i]$ is connected, and let $\Pi_G(V)\subseteq\Pi(V)$ denote the set of contiguous graph partitions.

We use two graph-informed constructions. The first compares the full edge partition induced by $P$, with the between-cluster edge class distinguished. Let $E^\star=E\cup\{\star\}$, where $\star\notin E$ is an auxiliary element added to identify the between-cluster edge class, and define the marked edge partition induced by $P$ as 
\[
\map(P)=\{E_0(P)\cup\{\star\}\}\cup\{E_i(P):1\leq i\leq K,\ E_i(P)\neq\varnothing\}\in\Pi(E^\star).
\]

\begin{definition}[Graph-informed distance]
For vertex partitions $P,Q\in\Pi(V)$ and a metric $\phi=\phi_{E^\star}$ on $\Pi(E^\star)$, define
$
\dG(P,Q;\phi)=\phi\{\map(P),\map(Q)\}.
$
\end{definition}

\begin{theorem}[Valid graph-informed distance metric]\label{thm:metric}
If $|E|>0$, then $\dG(\cdot,\cdot;\phi)$ is a distance metric on $\Pi_G(V)$.
\end{theorem}

\begin{remark}
On the full space $\Pi(V)$, the map $\dG(\cdot,\cdot;\phi)$ is in general only a pseudo-metric, since distinct non-contiguous vertex partitions can induce the same marked edge partition.
\end{remark}

\noindent\textit{Examples.}
Taking $X=E^\star$, two choices of $\phi$ are the variation of information of \citet{meilua2007comparing},
\[
\phi_{\mathrm{VI}}(R,S)=\HH(R)+\HH(S)-2\II(R,S),
\]
and the van Dongen metric of \citet{dongen2000performance},
\[
\begin{aligned}
\phi_{\mathrm{VD}}(R,S)
={}&2-|X|^{-1}\biggl\{
\sum_i\max_j|R_i\cap S_j|\\
&\hspace{37mm}+\sum_j\max_i|R_i\cap S_j|
\biggr\}.
\end{aligned}
\]
Here
\[
\HH(R)=-\sum_i \frac{|R_i|}{|X|}\log\left(\frac{|R_i|}{|X|}\right)
\]
and
\[
\II(R,S)=\sum_{i,j}\frac{|R_i\cap S_j|}{|X|}
\log\left\{\frac{|R_i\cap S_j|\,|X|}{|R_i|\,|S_j|}\right\},
\]
with $0\log0=0$. The corresponding graph-informed distances are denoted by $\dVI$ and $\dVD$.

The second construction is a binary cut-based companion construction that records only whether an edge is between clusters. Let $\mathcal{P}(E)=\{S:S\subseteq E\}$, let $\rho:\mathcal{P}(E)\times\mathcal{P}(E)\to\mathbb{R}_{\geq0}$ be a distance metric on $\mathcal{P}(E)$, and let $\cutmap(P)=E_0(P)$.

\begin{definition}[Cut-based graph-informed distance]
For vertex partitions $P,Q\in\Pi(V)$, define
$
\dG(P,Q;\rho)=\rho\{\cutmap(P),\cutmap(Q)\}.
$
\end{definition}

\begin{theorem}[Valid cut-based graph-informed distance metric]\label{thm:cutmetric}
If $|E|>0$, then $\dG(\cdot,\cdot;\rho)$ is a distance metric on $\Pi_G(V)$.
\end{theorem}

\noindent\textit{Example.}
A choice of $\rho$ is the Mirkin distance \citep{mirkin1970measurement}, $\rho(S,T)=|S\triangle T|$, where $\triangle$ denotes the symmetric difference for $S,T\in\mathcal{P}(E)$. 
This leads to a graph-informed Mirkin distance,
$
\dMK(P,Q)=|E_0(P)\triangle E_0(Q)|
$, which is also the Hamming distance between the binary edge labels indicating within-block versus between-block status.

\section{Graph-aware properties of graph-informed metrics}\label{sec:properties}
We investigate the graph-aware properties of the proposed distances from two perspectives. First, we formulate a local graph-aware refinement criterion for comparing one-step refinements of a fixed contiguous partition and show that the graph-informed distances of Section~\ref{sec:construction} admit corresponding local refinement scores. We then establish related graph-aware properties in inter-community and intra-community split settings under stochastic block models.

\subsection{Local refinement criterion}
We compare the distance between a fixed graph partition and its local refinements by examining how the edge profile is recorded when one block is split into two connected pieces. Let $A\in\Pi_G(V)$. A partition $B\in\Pi_G(V)$ is called a one-step refinement of $A$ if there exist a block $C(B)\in A$ and two non-empty disjoint sets $C_1(B),C_2(B)\subset C(B)$ with $C(B)=C_1(B)\cup C_2(B)$ such that
$
B=(A\setminus\{C(B)\})\cup\{C_1(B),C_2(B)\}.
$
Let $\mathcal R(A)$ denote the collection of all such refinements. For any $C\subseteq V$, write $E(C)=\{\{u,v\}\in E:u,v\in C\}$. For $B\in\mathcal R(A)$, define
$
\mu(B)=|E(C(B))|.
$
Since $A\in\Pi_G(V)$ and $C(B)$ is split into two non-empty sets, $\mu(B)>0$. Let $e_{11}(B)$, $e_{22}(B)$ and $e_{12}(B)$ denote the numbers of edges in $E(C(B))$ with both endpoints in $C_1(B)$, both endpoints in $C_2(B)$, and one endpoint in each of $C_1(B)$ and $C_2(B)$, respectively. Thus
$
e_{11}(B)+e_{22}(B)+e_{12}(B)=\mu(B).
$
Define $\alpha_{11}(B)=e_{11}(B)/\mu(B)$, $\alpha_{22}(B)=e_{22}(B)/\mu(B)$ and $\alpha_{12}(B)=e_{12}(B)/\mu(B)$, and write $\alpha(B)=(\alpha_{11}(B),\alpha_{22}(B),\alpha_{12}(B))$. Swapping the labels of the two child blocks exchanges only the first two coordinates.

\begin{definition}[Local graph-aware refinement preference]
Fix $A\in\Pi_G(V)$. A distance $d$ on $\Pi_G(V)$ is said to satisfy the local graph-aware refinement preference criterion if there exists a score function $F_A:\mathcal R(A)\to\mathbb{R}$ such that, for every $B\in\mathcal R(A)$, the value $F_A(B)$ depends on $B$ only through $\mu(B)$ and $\alpha(B)$, up to relabelling of the two child blocks, and for any $B_1,B_2\in\mathcal R(A)$,
$
F_A(B_1)\le F_A(B_2)\iff d(A,B_1)\le d(A,B_2).
$
\end{definition}

This criterion ensures that, relative to a fixed base partition, the ordering of a one-step refinement induced by the distance depends only on the edge mass inside the block being split and by how that mass is redistributed by the split.

A metric $\phi$ on $\Pi(E^\star)$ is called permutation-invariant if $\phi(R,S)=\phi\{\sigma(R),\sigma(S)\}$ for every permutation $\sigma$ of $E^\star$ and every $R,S\in\Pi(E^\star)$. A metric $\rho$ on $\mathcal P(E)$ is called permutation-invariant if $\rho(U,W)=\rho\{\sigma(U),\sigma(W)\}$ for every permutation $\sigma$ of $E$ and every $U,W\in\mathcal P(E)$. The variation of information and van Dongen metrics are permutation-invariant on $\Pi(E^\star)$, and the Mirkin metric is permutation-invariant on $\mathcal P(E)$.

\begin{theorem}[Existence of local refinement scores]\label{thm:localscore}
Fix $A\in\Pi_G(V)$. Let $d$ be either $\dG(\cdot,\cdot;\phi)$ for a permutation-invariant metric $\phi$ on $\Pi(E^\star)$, or $\dG(\cdot,\cdot;\rho)$ for a permutation-invariant metric $\rho$ on $\mathcal P(E)$. Then $d$ satisfies the local graph-aware refinement preference criterion on $\mathcal R(A)$.
\end{theorem}

\noindent\textit{Examples.}
Let $\eta_A=|E_0(A)|+1$, and let $h_2(t)=-t\log t-(1-t)\log(1-t)$. For the graph-informed Mirkin distance,
$
F_A^{(\kappa)}(B)=e_{12}(B)=\mu(B)\alpha_{12}(B).
$
For the graph-informed variation of information distance,
\[
F_A^{\mathrm{VI}}(B)=\mu(B)\HH\{\alpha(B)\}+\{\eta_A+e_{12}(B)\}h_2\!\left\{\frac{e_{12}(B)}{\eta_A+e_{12}(B)}\right\},
\]
where $\HH\{\alpha(B)\}=-\alpha_{11}(B)\log\alpha_{11}(B)-\alpha_{22}(B)\log\alpha_{22}(B)-\alpha_{12}(B)\log\alpha_{12}(B)$. For the graph-informed van Dongen distance,
\[
F_A^{\mathrm{VD}}(B)=\eta_A+\mu(B)+e_{12}(B)-\max\{e_{11}(B),e_{22}(B),e_{12}(B)\}-\max\{\eta_A,e_{12}(B)\}.
\]

Although the three scores appear dissimilar, each depends on the refinement only through the split profile $(\mu(B),e_{11}(B),e_{22}(B),e_{12}(B))$ together with the fixed quantity $\eta_A$, and each aggregates this profile in the manner of a distinct classical family of partition comparison. The Mirkin score retains only the cut coordinate $e_{12}(B)$, charging for severed edges alone and remaining indifferent to how the surviving mass is apportioned between the two children.
The variation of information score is an entropy functional of the full profile: the term
$\mu(B)\HH\{\alpha(B)\}$ is maximized, for fixed $\mu(B)$, by a balanced
three-way redistribution of the within-block mass and therefore assigns larger
values to stronger internal reorganization, while
$\{\eta_A+e_{12}(B)\}h_2\{e_{12}(B)/(\eta_A+e_{12}(B))\}$ records the
conditional uncertainty created when the newly cut edges are merged with the
pre-existing marked cut block. The van Dongen score is a dominant-overlap functional; when $\eta_A\geq e_{12}(B)$ it reduces to $\mu(B)+e_{12}(B)-\max\{e_{11}(B),e_{22}(B),e_{12}(B)\}$, retaining only the largest surviving fragment and penalizing near-even redistributions most heavily. A single split is thus scored through three complementary lenses -- cut mass, entropy of redistribution, and loss of dominant overlap.

\subsection{Inter-community split awareness}
Let $G=(V,E)$ follow a three-block SBM \citep{holland1983stochastic, abbe2018community} with $V=g_1\cup g_2\cup g_3$, $g_i\cap g_j=\varnothing$ for $i\neq j$, $|g_i|=n_i$, $N=n_1+n_2+n_3$, and $n_i/N\to \lambda_i\in(0,1)$ with $\lambda_1+\lambda_2+\lambda_3=1$. Conditional on the block memberships, edges are independent with probability $p_i$ for pairs inside $g_i$ and $q_{ij}$ for pairs between $g_i$ and $g_j$, where $0<p_i<1$ and $0<q_{ij}<1$. We compare the one-block partition $A=\{V\}$ with $B_1=\{g_1,\ g_2\cup g_3\}$ and $B_2=\{g_3,\ g_1\cup g_2\}$. For compactness, write $W_1=\lambda_1^2p_1/2$, $W_3=\lambda_3^2p_3/2$, $U_{23}=\lambda_2^2p_2/2+\lambda_3^2p_3/2+\lambda_2\lambda_3q_{23}$, $U_{12}=\lambda_1^2p_1/2+\lambda_2^2p_2/2+\lambda_1\lambda_2q_{12}$, $C_{1\cdot}=\lambda_1\lambda_2q_{12}+\lambda_1\lambda_3q_{13}$, $C_{\cdot 3}=\lambda_2\lambda_3q_{23}+\lambda_1\lambda_3q_{13}$, and $T=W_1+\lambda_2^2p_2/2+W_3+\lambda_1\lambda_2q_{12}+\lambda_2\lambda_3q_{23}+\lambda_1\lambda_3q_{13}$.

\begin{theorem}[Inter-community split awareness]\label{thm:inter}
Under the three-block model above, the following almost sure limits hold.

\noindent(i) If $\lambda_1q_{12}>\lambda_3q_{23}$, then
$
N^{-2}\{\dMK(A,B_1)-\dMK(A,B_2)\}\to \delta_{\kappa}^{\mathrm{inter}}>0
$
almost surely.

\noindent(ii) If $U_{23}>\max\{C_{1\cdot},W_1\}$, $U_{12}>\max\{C_{\cdot 3},W_3\}$, $\lambda_1^2p_1 \geq \lambda_3^2p_3$, and $\lambda_1q_{12}>\lambda_3q_{23}$, then
$
\dVI(A,B_1)-\dVI(A,B_2)\to \delta_{\mathrm{VI}}^{\mathrm{inter}}>0
$
almost surely.

\noindent(iii) If $U_{23}>\max\{C_{1\cdot},W_1\}$, $U_{12}>\max\{C_{\cdot 3},W_3\}$, and $W_1+\lambda_1\lambda_2q_{12}>W_3+\lambda_2\lambda_3q_{23}$, then
$
\dVD(A,B_1)-\dVD(A,B_2)\to \delta_{\mathrm{VD}}^{\mathrm{inter}}>0
$
almost surely.
\end{theorem}

\begin{table}
\centering
\footnotesize
\caption{Almost sure limits in Theorem~\ref{thm:inter} and Corollary~\ref{cor:inter-eq}. The general column gives the limit under the three-block model; the equal-sized column specializes to $\lambda_1=\lambda_2=\lambda_3=1/3$ and $p_1=p_2=p_3=p$, with $D:=D_{\mathrm{eq},3}=3p+2(q_{12}+q_{23}+q_{13})$.}
\label{tab:inter-limits}
\begin{tabular}{p{0.16\textwidth}p{0.40\textwidth}p{0.36\textwidth}}
\hline
Distance & Limit (general) & Limit (equal-sized) \\
\hline
Mirkin &
$\delta_{\kappa}^{\mathrm{inter}}=\lambda_2(\lambda_1q_{12}-\lambda_3q_{23})$ &
$\delta_{\kappa}^{\mathrm{inter,eq}}=(q_{12}-q_{23})/9$ \\[6pt]

Variation of information &
$\begin{aligned}[t]
\delta_{\mathrm{VI}}^{\mathrm{inter}}
={}&
\HH(W_1/T,U_{23}/T,C_{1\cdot}/T)\\
&-\HH(W_3/T,U_{12}/T,C_{\cdot3}/T)
\end{aligned}$ &
$\begin{aligned}[t]
\delta_{\mathrm{VI}}^{\mathrm{inter,eq}}
={}&
\HH\{\tfrac{p}{D},\tfrac{2(p+q_{23})}{D},\tfrac{2(q_{12}+q_{13})}{D}\}\\
&-\HH\{\tfrac{p}{D},\tfrac{2(p+q_{12})}{D},\tfrac{2(q_{23}+q_{13})}{D}\}
\end{aligned}$ \\[18pt]

van Dongen &
$\delta_{\mathrm{VD}}^{\mathrm{inter}}=(U_{12}-U_{23})/T$ &
$\delta_{\mathrm{VD}}^{\mathrm{inter,eq}}=2(q_{12}-q_{23})/D$ \\
\hline
\end{tabular}
\end{table}

Theorem~\ref{thm:inter} shows that the three distances reflect different local manifestations of the same structural comparison between the two refinements. For the Mirkin distance, the only relevant quantity is the excess cut load created by separating $g_1$ from $g_2\cup g_3$ rather than $g_3$ from $g_1\cup g_2$. For variation of information, the comparison depends on the full three-way marked edge profile induced by each refinement, and positivity follows when mass is transferred from the outer edge classes to the dominant merged-block class. For van Dongen distance, the dominance conditions identify the asymptotically largest induced marked edge block under each refinement, after which the ordering is determined by the relative sizes of the two merged-block masses. The assumptions in Theorem~\ref{thm:inter} have simple structural interpretations. For the graph-informed Mirkin distance, the condition $\lambda_1 q_{12}>\lambda_3 q_{23}$ says that, after weighting by block sizes, separating $g_1$ from $g_2$ destroys more cross-community connectivity than separating $g_2$ from $g_3$, so $B_1$ incurs the larger cut load. For graph-informed variation of information, the dominance conditions $U_{23}>\max\{C_{1\cdot},W_1\}$ and $U_{12}>\max\{C_{\cdot 3},W_3\}$ ensure that the merged-block marked edge class is asymptotically the largest under each refinement, while $\lambda_1^2 p_1\geq \lambda_3^2 p_3$ together with $\lambda_1 q_{12}>\lambda_3 q_{23}$ imply that $B_1$ has a less concentrated three-way edge profile than $B_2$, and hence the larger entropy-based distance. For graph-informed van Dongen distance, the same dominance conditions again identify the merged-block marked edge class as the maximizer, and $W_1+\lambda_1\lambda_2 q_{12}>W_3+\lambda_2\lambda_3 q_{23}$ states that the merged side under $B_2$ is asymptotically larger than that under $B_1$, which determines the ordering.

\begin{corollary}[Equal-sized inter-community case]\label{cor:inter-eq}
Assume the setting of Theorem~\ref{thm:inter} with $\lambda_1=\lambda_2=\lambda_3=1/3$ and $p_1=p_2=p_3=p$. If $q_{12}>q_{23}$, $p+q_{23}>q_{12}+q_{13}$, and $p+q_{12}>q_{23}+q_{13}$, then the conclusions of Theorem~\ref{thm:inter} hold for all three graph-informed distances. More precisely,
$
N^{-2}\{\dMK(A,B_1)-\dMK(A,B_2)\}\to \delta_{\kappa}^{\mathrm{inter,eq}}>0,
$
$
\dVI(A,B_1)-\dVI(A,B_2)\to \delta_{\mathrm{VI}}^{\mathrm{inter,eq}}>0,
$
and
$
\dVD(A,B_1)-\dVD(A,B_2)\to \delta_{\mathrm{VD}}^{\mathrm{inter,eq}}>0
$
almost surely.
\end{corollary}

\begin{remark}
Corollary~\ref{cor:inter-eq} gives the clean interpretation under the equal-sized case. The refinements $B_1=\{g_1,\ g_2\cup g_3\}$ and $B_2=\{g_3,\ g_1\cup g_2\}$ have the same vertex-level form relative to $A=\{V\}$: each isolates one group and merges the other two. Any graph-unaware distance depending only on block sizes or label assignments therefore assigns the same distance from $A$ to $B_1$ and $B_2$. By contrast, the graph-informed distances detect that the cut separating $g_1$ from $g_2$ is stronger than the cut separating $g_2$ from $g_3$ when $q_{12}>q_{23}$, and hence rank $B_1$ as farther from $A$. Fig.~\ref{fig:inter} provides a visual summary of the inter-community comparison.
\end{remark}

\begin{figure}
\centering
\begin{subfigure}[t]{0.48\textwidth}
\centering
\includegraphics[width=\textwidth]{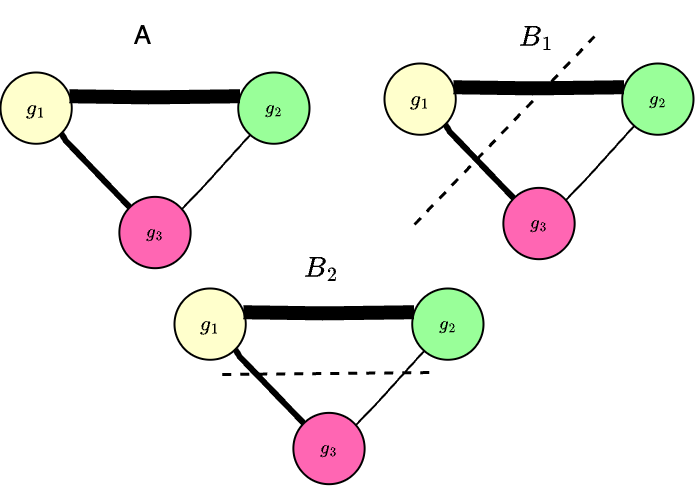}
\caption{Three-block (inter-community) setting.}
\label{fig:inter}
\end{subfigure}
\hfill
\begin{subfigure}[t]{0.48\textwidth}
\centering
\includegraphics[width=\textwidth]{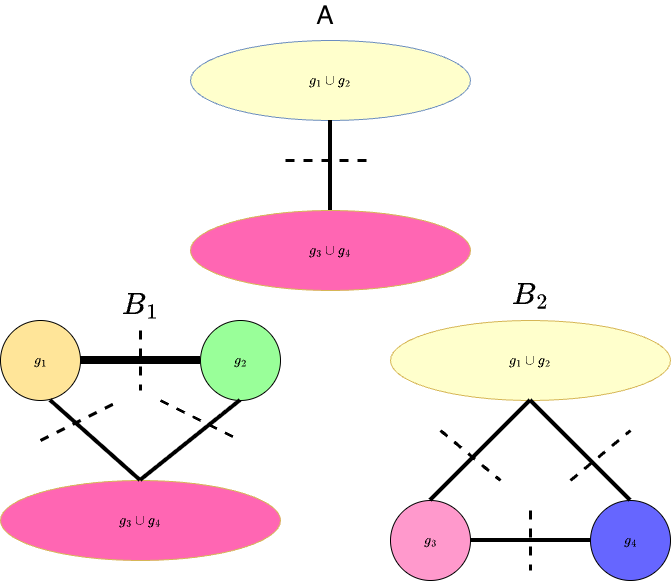}
\caption{Two-block (intra-community) setting.}
\label{fig:intra}
\end{subfigure}
\caption{Schematic representations of the graph partitions $A$, $B_1$, and $B_2$.
(a) Three-block setting: the dotted lines indicate the cuts producing $B_1=\{g_1,\ g_2\cup g_3\}$ and $B_2=\{g_3,\ g_1\cup g_2\}$, with the $g_1$--$g_2$ connection drawn thicker than the $g_2$--$g_3$ connection, reflecting $q_{12}>q_{23}$, under which the graph-informed distances rank $B_1$ farther from $A$ than $B_2$.
(b) Two-block setting: the dotted lines indicate the cuts producing $A=\{g_1\cup g_2,\ g_3\cup g_4\}$; $B_1=\{g_1,\ g_2,\ g_3\cup g_4\}$; $B_2=\{g_1\cup g_2,\ g_3,\ g_4\}$. The within-block connectivity on $g_1\cup g_2$ is drawn thicker than that on $g_3\cup g_4$, corresponding to $p_1>p_2$.}
\label{fig:sbm}
\end{figure}

\subsection{Intra-community split awareness}

Let $G=(V,E)$ follow a two-block stochastic block model with meta communities 
$L := g_1 \cup g_2$ and $R := g_3 \cup g_4$, 

where $|g_1|=|g_3|=n_1$, $|g_2|=|g_4|=n_2$, $N=2n_1+2n_2$, and $r=n_1/n_2\in(0,\infty)$ fixed as $n_1,n_2\to\infty$. Conditional on the block memberships, edges are independent with probability $p_1$ for pairs inside $g_1\cup g_2$, probability $p_2$ for pairs inside $g_3\cup g_4$, and probability $q$ for pairs with one endpoint in $g_1\cup g_2$ and the other in $g_3\cup g_4$, where $0<p_1,p_2,q<1$. 

We compare $A=\{g_1\cup g_2,\ g_3\cup g_4\}$, $B_1=\{g_1,\ g_2,\ g_3\cup g_4\}$, and $B_2=\{g_1\cup g_2,\ g_3,\ g_4\}$ to study how graph-informed distances order different intra-community refinements.

For compactness, let $K(r)=r+r^{-1}+2$, $S_1(r)=1+\{r+r^{-1}\}/2$, $M_1(r)=\max\{r/2,\ 1/(2r),\ 1\}$, $\Phi(r)=2+\{r+r^{-1}\}/2-M_1(r)$, and $D(r)=(p_1+p_2)S_1(r)+qK(r)$. Also let $u_0(r)=S_1(r)^{-1}$, $u_1(r)=r\{2S_1(r)\}^{-1}$, $u_2(r)=\{2rS_1(r)\}^{-1}$, let $h_2(x)=-x\log x-(1-x)\log(1-x)$, let $h_3(r)=-\sum_{j=0}^2u_j(r)\log u_j(r)$, and define $F_r(p)=\{p+qK(r)\}h_2[p/\{p+qK(r)\}]+pS_1(r)h_3(r)$. All almost sure statements are with respect to a sequence of graphs generated under the above model as $N\to\infty$.

\begin{theorem}[Intra-community split awareness]\label{thm:intra}
Under the two-block model above, the following almost sure limits hold.

\noindent(i) If $p_1>p_2$, then
$
N^{-2}\{\dMK(A,B_1)-\dMK(A,B_2)\}\to \{r/4(r+1)^2\}(p_1-p_2)>0
$
almost surely.

\noindent(ii) If $p_1>p_2$, then
$
\dVI(A,B_1)-\dVI(A,B_2)\to \{F_r(p_1)-F_r(p_2)\}/D(r)>0
$
almost surely.

\noindent(iii) If $p_1>p_2$ and $p_2>qK(r)$, then
$
\dVD(A,B_1)-\dVD(A,B_2)\to \{(p_1-p_2)(\Phi(r)-1)\}/D(r)>0
$
almost surely.
\end{theorem}

In the two-block SBM setting, we are essentially comparing two
size-matched refinements of an underlying two-community
graph on the meta-communities $L = g_1 \cup g_2$ and $R = g_3 \cup
g_4$: the refinement $B_1$ refines $L$, and $B_2$ refines $R$. When
$p_1 > p_2$, the meta-community $L$ is denser than $R$, so $B_1$ is
the refinement that cuts the denser of the two meta-communities while
$B_2$ cuts the sparser one. 
Theorem~\ref{thm:intra} states that all
three graph-informed distances are asymptotically and almost surely
strictly larger from $A$ to $B_1$ than from $A$ to $B_2$ under the conditions of Theorem~\ref{thm:intra}: a
size-matched refinement of a denser community is more topologically
disruptive than the same refinement of a sparser community, and the
graph-informed distances correctly reflect this. Any distance
depending only on vertex-level block sizes or label assignments, by
contrast, treats $B_1$ and $B_2$ identically in this size-matched regime: at the vertex level, the two refinements have the same
form relative to $A$.

Theorem \ref{thm:intra} also shows that the difference in the Mirkin distance depends only on the extra cut mass produced by splitting $g_1\cup g_2$ rather than $g_3\cup g_4$. The variation of information limit depends on a one-dimensional contrast through the strictly increasing function $F_r$, so the ordering is governed directly by the inequality $p_1>p_2$. The van Dongen conclusion again depends on the dominant induced edge class; the additional condition $p_2>qK(r)$ ensures that the within-side classes dominate the cross-side class on both sides, after which the sign is controlled by $(p_1-p_2)$. 

\begin{corollary}[Equal-sized intra-community case]\label{cor:intra-eq}
Assume the setting of Theorem~\ref{thm:intra} with $n_1=n_2$, so that $r=1$, $K(1)=4$, $\Phi(1)=2$, and $D(1)=2(p_1+p_2+2q)$. If $p_1>p_2$, then\\
$N^{-2}\{\dMK(A,B_1)-\dMK(A,B_2)\}\to (p_1-p_2)/16>0$ almost surely,\\
$\dVI(A,B_1)-\dVI(A,B_2)\to \{F_1(p_1)-F_1(p_2)\}/D(1)>0$ almost surely, where $F_1(p)=(p+4q)\,h_2\{p(p+4q)^{-1}\}+3p\log 2$,\\
and, if in addition $p_2>4q$, $\dVD(A,B_1)-\dVD(A,B_2)\to (p_1-p_2)/D(1)>0$ almost surely.
\end{corollary}

\begin{remark}
Corollary~\ref{cor:intra-eq} gives the corresponding interpretation for within-community splitting when each meta-community is split into two equal-sized pieces. When $n_1=n_2$, the two refinements $B_1=\{g_1,\ g_2,\ g_3\cup g_4\}$ and $B_2=\{g_1\cup g_2,\ g_3,\ g_4\}$ have the same vertex-level form relative to $A=\{g_1\cup g_2,\ g_3\cup g_4\}$: each splits one of two equal-sized meta-communities into two equal parts. A graph-unaware distance therefore treats them identically. The graph-informed distances, however, distinguish them through edge density, and correctly assign a larger distance to $B_1$ when the side $g_1\cup g_2$ is denser than $g_3\cup g_4$, that is, when $p_1>p_2$. Fig.~\ref{fig:intra} gives the analogous picture for intra-community splitting.
\end{remark}

\section{Simulation study}

We illustrate the finite-sample behavior of the three graph-informed distances in the equal-sized settings of Corollaries~\ref{cor:inter-eq} and \ref{cor:intra-eq}. In the three-block case, we set $p_1=p_2=p_3=0.8$, $q_{13}=0.05$, and $q_{23}=0.10$, then fix $q_{12}=0.30$ and vary the graph size $N$, and also fix $N=600$ and vary the signal gap $q_{12}-q_{23}$. In the two-block case, we set $p_1=0.8$, $p_2=0.4$, and $q=0.05$ and vary $N$, and also fix $N=800$, $p_2=0.3$, $q=0.05$ and vary the signal gap $p_1-p_2$. For each parameter setting, we generated 50 independent graphs and computed the Monte Carlo mean of $d_G(A,B_1)-d_G(A,B_2)$ for the graph-informed Mirkin distance, normalized by $N^2$, together with the graph-informed variation of information and van Dongen distances.

\begin{figure}
\centering
\includegraphics[width=0.75\textwidth]{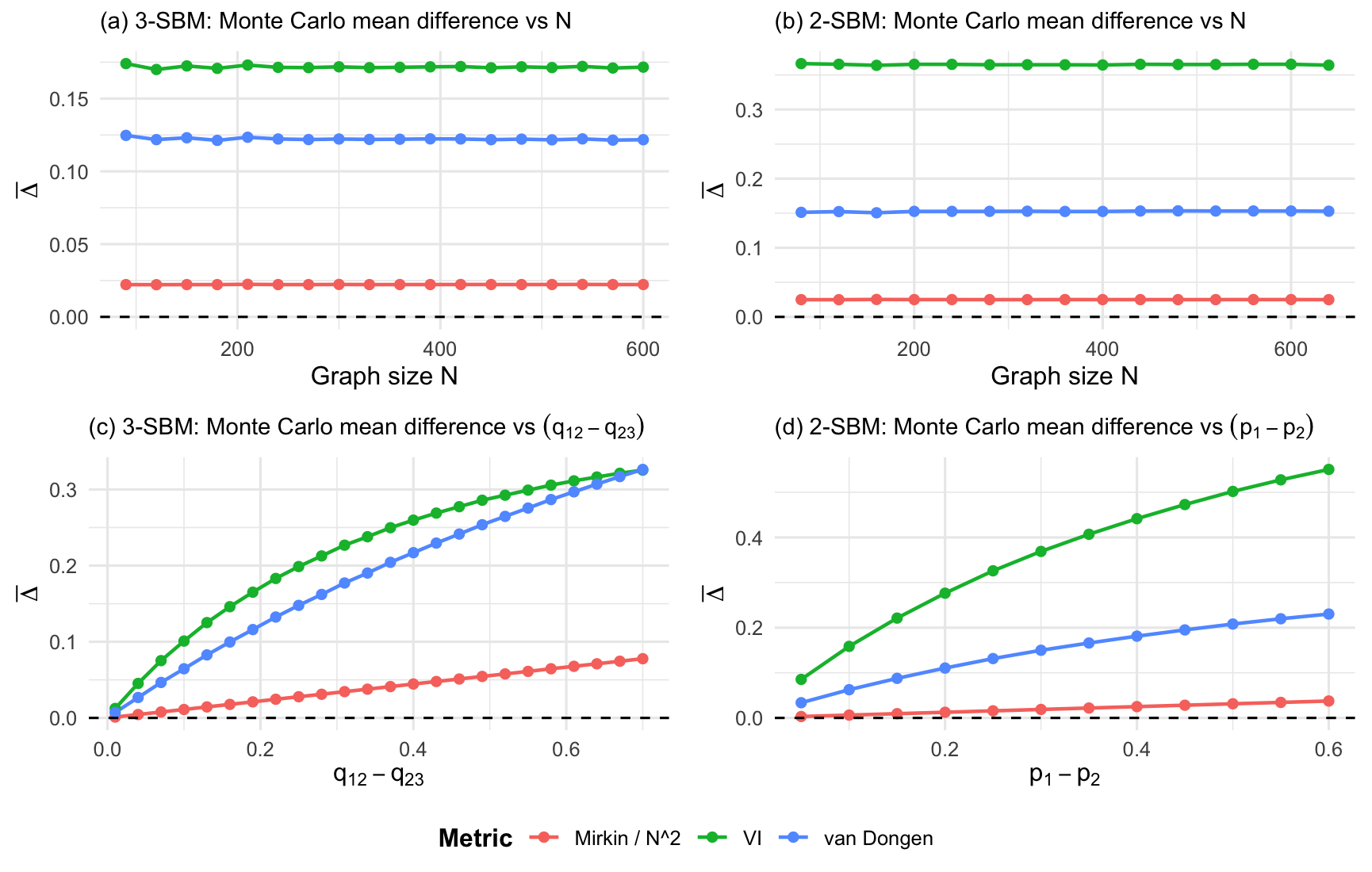}
\caption{Monte Carlo mean differences $d_G(A,B_1)-d_G(A,B_2)$ for the graph-informed Mirkin, variation of information, and van Dongen distances. Panels (a) and (b) vary graph size in the equal-sized three-block and two-block settings. Panels (c) and (d) fix graph size and vary $q_{12}-q_{23}$ and $p_1-p_2$. Positive values indicate that $B_1$ is farther from $A$ than $B_2$.}
\label{fig:comparison}
\end{figure}

Fig.~\ref{fig:comparison} supports the theory. In panels (a) and (b), the Monte Carlo means are positive and stabilize quickly as the graph size increases, consistent with the asymptotic ordering $d_G(A,B_1)>d_G(A,B_2)$ established in the Corollaries. In panels (c) and (d), the gap, $d_G(A,B_1)-d_G(A,B_2)$, widens monotonically as the relevant structural contrast, $q_{12}-q_{23}$ (inter-community) and $p_1-p_2$ (intra-community), strengthens, confirming that the distances respond smoothly to the magnitude of the disruption rather than merely to its presence. The Mirkin curve is consistently the smallest, since it registers only the change in cut mass, whereas the variation of information and van Dongen distances react more strongly because they incorporate the full induced edge-partition structure. Across all four panels the ordering is the same: $B_1$ is systematically farther from $A$ than $B_2$ whenever the graph structure makes the split underlying $B_1$ more disruptive, illustrating that all three graph-informed distances recover the topology-aware ranking that graph-agnostic distances cannot.

\section{Conclusion}
Graph-informed distances compare partitions through the edge structure they induce and thereby distinguish refinements that are topologically different but combinatorially similar at the vertex level. The metric construction in Section~\ref{sec:construction} gives a general route from distances between set partitions to distances between contiguous graph partitions, while the local refinement criterion characterizes how graph-aware distance reflects graph structures through the edge mass and split profile created by a refinement. The stochastic block model results show that graph-informed Mirkin, variation of information, and van Dongen distances all respond correctly to intuitively stronger inter-community and intra-community disruptions, though through different local mechanisms. An important direction for future work is to extend the present framework to weighted, directed, and dynamic graphs, where partition comparisons must account not only for connectivity but also for edge heterogeneity and temporal evolution.

\bibliographystyle{plainnat}
\bibliography{paper-ref}

\newpage
\appendix
\section{Proof of Theorem \ref{thm:metric}}
\begin{proof}
Fix $P,Q,R\in\Pi_G(V)$. Since $\map(P),\map(Q),\map(R)\in\Pi(E^\star)$ and $\phi$ is a distance metric on $\Pi(E^\star)$, the quantity
$
\dG(P,Q;\phi)=\phi\{\map(P),\map(Q)\}
$
is well-defined. Non-negativity, symmetry and the triangle inequality follow directly from the corresponding properties of $\phi$. Indeed,
\[
\dG(P,Q;\phi)=\phi\{\map(P),\map(Q)\}\geq0,
\]
\[
\dG(P,Q;\phi)=\phi\{\map(P),\map(Q)\}
=\phi\{\map(Q),\map(P)\}=\dG(Q,P;\phi),
\]
and
\[
\dG(P,R;\phi)=\phi\{\map(P),\map(R)\}
\leq \phi\{\map(P),\map(Q)\}+\phi\{\map(Q),\map(R)\}
=\dG(P,Q;\phi)+\dG(Q,R;\phi).
\]

It remains to prove the identity of indiscernibles on $\Pi_G(V)$. If $P=Q$, then $\map(P)=\map(Q)$ and hence $\dG(P,Q;\phi)=0$. Conversely, suppose $\dG(P,Q;\phi)=0$. Since $\phi$ is a metric on $\Pi(E^\star)$, this implies
\[
\map(P)=\map(Q)
\]
as partitions of $E^\star$. In $\map(P)$, the unique block containing $\star$ is $E_0(P)\cup\{\star\}$, while in $\map(Q)$ the unique block containing $\star$ is $E_0(Q)\cup\{\star\}$. Therefore equality of the marked edge partitions gives
\[
E_0(P)\cup\{\star\}=E_0(Q)\cup\{\star\},
\]
and since $\star\notin E$, it follows that
\[
E_0(P)=E_0(Q).
\]

We now show that this equality implies $P=Q$. Let
\[
H_P=(V,E\setminus E_0(P)).
\]
For any block $P_i$ of $P$, all edges of $G[P_i]$ belong to $E\setminus E_0(P)$, and since $P\in\Pi_G(V)$, the subgraph $G[P_i]$ is connected. Hence all vertices in $P_i$ lie in a single connected component of $H_P$. Conversely, no edge of $E\setminus E_0(P)$ has endpoints in two different blocks of $P$, because every such edge would belong to $E_0(P)$. Thus no connected component of $H_P$ can contain vertices from two different blocks of $P$. Therefore the connected components of $H_P$ are exactly the blocks of $P$.

The same argument applied to $Q$ shows that the connected components of
\[
H_Q=(V,E\setminus E_0(Q))
\]
are exactly the blocks of $Q$. Since $E_0(P)=E_0(Q)$, we have $H_P=H_Q$, so their connected components are identical. Hence $P=Q$. Therefore $\dG(P,Q;\phi)=0$ implies $P=Q$ on $\Pi_G(V)$, and $\dG(\cdot,\cdot;\phi)$ is a distance metric on $\Pi_G(V)$.
\end{proof}

\section{Proof of Theorem \ref{thm:cutmetric}}
\begin{proof}
Fix $P,Q,R\in\Pi_G(V)$. Since $\cutmap(P)=E_0(P)$, $\cutmap(Q)=E_0(Q)$ and $\cutmap(R)=E_0(R)$ are subsets of $E$, and since $\rho$ is a distance metric on $\mathcal{P}(E)$, the quantity
$
\dG(P,Q;\rho)=\rho\{\cutmap(P),\cutmap(Q)\}
$
is well-defined. Non-negativity, symmetry and the triangle inequality follow directly from the corresponding properties of $\rho$:
\[
\dG(P,Q;\rho)=\rho\{\cutmap(P),\cutmap(Q)\}\geq0,
\]
\[
\dG(P,Q;\rho)=\rho\{\cutmap(P),\cutmap(Q)\}
=\rho\{\cutmap(Q),\cutmap(P)\}=\dG(Q,P;\rho),
\]
and
\[
\dG(P,R;\rho)=\rho\{\cutmap(P),\cutmap(R)\}
\leq \rho\{\cutmap(P),\cutmap(Q)\}+\rho\{\cutmap(Q),\cutmap(R)\}
=\dG(P,Q;\rho)+\dG(Q,R;\rho).
\]

It remains to prove the identity of indiscernibles on $\Pi_G(V)$. If $P=Q$, then $\cutmap(P)=\cutmap(Q)$ and hence $\dG(P,Q;\rho)=0$. Conversely, suppose $\dG(P,Q;\rho)=0$. Since $\rho$ is a metric on $\mathcal{P}(E)$, this implies
\[
\cutmap(P)=\cutmap(Q),
\]
or equivalently,
\[
E_0(P)=E_0(Q).
\]

As in the proof of Theorem~\ref{thm:metric}, the connected components of
\[
(V,E\setminus E_0(P))
\]
are exactly the blocks of $P$, because $P\in\Pi_G(V)$ and all edges between distinct blocks of $P$ are removed. Similarly, the connected components of
\[
(V,E\setminus E_0(Q))
\]
are exactly the blocks of $Q$. Since $E_0(P)=E_0(Q)$, these two graphs are identical and hence have the same connected components. Therefore $P=Q$.

Thus $\dG(P,Q;\rho)=0$ implies $P=Q$ for $P,Q\in\Pi_G(V)$, and $\dG(\cdot,\cdot;\rho)$ is a distance metric on $\Pi_G(V)$.
\end{proof}

\section{Proof of Theorem \ref{thm:localscore}}
\begin{proof}
Fix the graph $G=(V,E)$ and a base partition $A\in\Pi_G(V)$. For $B\in\mathcal R(A)$, let $C(B)\in A$ be the unique block of $A$ split by $B$, and write $C(B)=C_1(B)\cup C_2(B)$ with $C_1(B)$ and $C_2(B)$ non-empty and disjoint. Let
\[
E_{11}(B)=\{\{u,v\}\in E(C(B)):u,v\in C_1(B)\},
\]
\[
E_{22}(B)=\{\{u,v\}\in E(C(B)):u,v\in C_2(B)\},
\]
and
\[
E_{12}(B)=\{\{u,v\}\in E(C(B)):u\in C_1(B),\ v\in C_2(B)\}.
\]
These three sets are pairwise disjoint and their union is $E(C(B))$. Their cardinalities are $e_{11}(B)$, $e_{22}(B)$ and $e_{12}(B)$, respectively, and $\mu(B)=|E(C(B))|$.  

It is enough to prove that, whenever $B,B'\in\mathcal R(A)$ have the same local edge statistics, $(\mu, \alpha)$, up to relabelling of the two child blocks, the corresponding distances from $A$ are equal. After possibly interchanging the labels of $C_1(B')$ and $C_2(B')$, we may assume
\[
\mu(B)=\mu(B'),\qquad e_{11}(B)=e_{11}(B'),\qquad e_{22}(B)=e_{22}(B'),\qquad e_{12}(B)=e_{12}(B').
\]
Indeed, equality of $\mu$ and $\alpha$ gives equality of the three edge counts.

We first consider the marked edge-partition distance $\dG(\cdot,\cdot;\phi)$. Let
$
P_A=\map(A),\ P_B=\map(B),\ P_{B'}=\map(B')
$
be the marked edge partitions in $\Pi(E^\star)$. Since $B$ is obtained from $A$ by splitting $C(B)$ into $C_1(B)$ and $C_2(B)$, the marked between-cluster edge class of $P_B$ is
\[
E_0(B)\cup\{\star\}=E_0(A)\cup E_{12}(B)\cup\{\star\},
\]
the possible non-empty within-block edge classes $E_{11}(B)$ and $E_{22}(B)$ replace the class $E(C(B))$ in the within-cluster edge class of $P_B$, and every other within-block class of $P_A$ is unchanged. The same description holds for $B'$.

By permutation invariance of $\phi$, it is enough to construct a permutation $\sigma$ of $E^\star$ such that
\[
\sigma(P_A)=P_A,\qquad \sigma(P_B)=P_{B'}.
\]
We construct $\sigma$ in two cases.

\medskip
\noindent\textit{Case 1: $C(B)=C(B')$.}
Since the corresponding local edge counts are equal, there exist bijections
\[
f_{11}:E_{11}(B)\to E_{11}(B'),\qquad
f_{22}:E_{22}(B)\to E_{22}(B'),\qquad
f_{12}:E_{12}(B)\to E_{12}(B').
\]
Define $\sigma:E^\star\to E^\star$ by
\[
\sigma(e)=
\begin{cases}
f_{11}(e), & e\in E_{11}(B),\\
f_{22}(e), & e\in E_{22}(B),\\
f_{12}(e), & e\in E_{12}(B),\\
e, & e\in E^\star\setminus E(C(B)).
\end{cases}
\]
Empty sets cause no difficulty. Since $E_{11}(B)$, $E_{22}(B)$ and $E_{12}(B)$ partition $E(C(B))$, and the three maps above are bijections onto the corresponding classes for $B'$, $\sigma$ is a permutation of $E^\star$. The block $E(C(B))$ of $P_A$ is mapped onto itself, while every other block of $P_A$, including $E_0(A)\cup\{\star\}$, is fixed as a set. Hence $\sigma(P_A)=P_A$. Moreover, the blocks $E_{11}(B)$ and $E_{22}(B)$ are mapped onto $E_{11}(B')$ and $E_{22}(B')$, respectively, and
\[
\sigma\{E_0(A)\cup E_{12}(B)\cup\{\star\}\}
=
E_0(A)\cup E_{12}(B')\cup\{\star\}.
\]
Therefore $\sigma(P_B)=P_{B'}$.

\medskip
\noindent\textit{Case 2: $C(B)\neq C(B')$.}
The edge sets $E(C(B))$ and $E(C(B'))$ are disjoint, since $C(B)$ and $C(B')$ are distinct blocks of $A$. By the equality of local counts, there exist bijections
\[
f_{11}:E_{11}(B)\to E_{11}(B'),\qquad
f_{22}:E_{22}(B)\to E_{22}(B'),\qquad
f_{12}:E_{12}(B)\to E_{12}(B').
\]
Together these define a bijection $f:E(C(B))\to E(C(B'))$ that sends each local edge class of $B$ to the corresponding local edge class of $B'$. Let $f^{-1}:E(C(B'))\to E(C(B))$ be its inverse. Define $\sigma:E^\star\to E^\star$ by
\[
\sigma(e)=
\begin{cases}
f(e), & e\in E(C(B)),\\
f^{-1}(e), & e\in E(C(B')),\\
e, & e\in E^\star\setminus\{E(C(B))\cup E(C(B'))\}.
\end{cases}
\]
This is a permutation of $E^\star$. In $P_A$, the sets $E(C(B))$ and $E(C(B'))$ are two within-block edge classes; $\sigma$ swaps these two blocks and fixes every other block, including $E_0(A)\cup\{\star\}$. Hence $\sigma(P_A)=P_A$. In $P_B$, the block $E(C(B))$ has been replaced by the possible non-empty classes $E_{11}(B)$ and $E_{22}(B)$, while $E_{12}(B)$ has been added to the marked cut block. These are mapped to the corresponding classes for $B'$. Also, the unsplit class $E(C(B'))$ in $P_B$ is mapped to $E(C(B))$, which is the unsplit class in $P_{B'}$. Finally,
\[
\sigma\{E_0(A)\cup E_{12}(B)\cup\{\star\}\}
=
E_0(A)\cup E_{12}(B')\cup\{\star\}.
\]
Thus $\sigma(P_B)=P_{B'}$.

In both cases, permutation invariance of $\phi$ gives
\[
\dG(A,B;\phi)
=
\phi(P_A,P_B)
=
\phi\{\sigma(P_A),\sigma(P_B)\}
=
\phi(P_A,P_{B'})
=
\dG(A,B';\phi).
\]

We next consider the cut-based distance $\dG(\cdot,\cdot;\rho)$. Since $B$ is a one-step refinement of $A$,
\[
\cutmap(B)=E_0(B)=E_0(A)\cup E_{12}(B),
\]
and similarly $\cutmap(B')=E_0(A)\cup E_{12}(B')$. Since $e_{12}(B)=e_{12}(B')$, the two sets $E_{12}(B)$ and $E_{12}(B')$ have the same cardinality. Hence there exists a permutation $\tau$ of $E$ such that
\[
\tau\{E_0(A)\}=E_0(A),\qquad
\tau\{E_{12}(B)\}=E_{12}(B').
\]
For example, choose any bijection from $E_{12}(B)$ to $E_{12}(B')$, extend it to a permutation of $E\setminus E_0(A)$, and fix $E_0(A)$ pointwise. Then
\[
\tau\{\cutmap(A)\}=\cutmap(A),\qquad
\tau\{\cutmap(B)\}=\cutmap(B').
\]
By permutation invariance of $\rho$,
\[
\dG(A,B;\rho)
=
\rho\{\cutmap(A),\cutmap(B)\}
=
\rho\{\tau(\cutmap(A)),\tau(\cutmap(B))\}
=
\rho\{\cutmap(A),\cutmap(B')\}
=
\dG(A,B';\rho).
\]

Thus, for either construction, $d(A,B)$ is constant over all refinements $B\in\mathcal R(A)$ with the same local statistics $\{\mu(B),\alpha(B)\}$ up to relabelling of the two child blocks. Let $\mathcal S_A$ be the set of attainable local-statistic pairs $(\mu,\alpha)$, with the first two coordinates of $\alpha$ identified under relabelling. Define
\[
f_A:\mathcal S_A\to\mathbb{R},\qquad
f_A(\mu,\alpha)=d(A,B),
\]
where $B$ is any element of $\mathcal R(A)$ with local statistics $(\mu,\alpha)$. The preceding argument shows that $f_A$ is well-defined. Now set
\[
F_A(B)=f_A\{\mu(B),\alpha(B)\},\qquad B\in\mathcal R(A).
\]
Then $F_A(B)$ depends on $B$ only through $\mu(B)$ and $\alpha(B)$, up to relabelling of the two child blocks, and
\[
F_A(B)=d(A,B)
\]
for every $B\in\mathcal R(A)$. Therefore, for any $B_1,B_2\in\mathcal R(A)$,
\[
F_A(B_1)\le F_A(B_2)
\iff
d(A,B_1)\le d(A,B_2).
\]
Hence $d$ satisfies the local graph-aware refinement preference criterion on $\mathcal R(A)$.
\end{proof}

\section{Proof of Theorem \ref{thm:inter}}
\begin{proof}
For $1\leq i\leq 3$, let $\mathsf W_i$ denote the number of observed edges with both endpoints in $g_i$, and for $1\leq i<j\leq 3$, let $\mathsf C_{ij}$ denote the number of observed edges between $g_i$ and $g_j$. Write
\[
\mathsf M=|E|=\mathsf W_1+\mathsf W_2+\mathsf W_3+\mathsf C_{12}+\mathsf C_{23}+\mathsf C_{13}.
\]
Since the edge indicators are independent and the block memberships are deterministic, the strong law of large numbers gives
\[
N^{-2}\mathsf W_i \to \lambda_i^2p_i/2 \quad \text{a.s.}\qquad (i=1,2,3),
\]
and
\[
N^{-2}\mathsf C_{ij}\to \lambda_i\lambda_jq_{ij}\quad \text{a.s.}\qquad (1\leq i<j\leq 3).
\]
Hence, on an event of probability one that we fix throughout the proof,
\[
N^{-2}\mathsf M\to T=W_1+\lambda_2^2p_2/2+W_3+\lambda_1\lambda_2q_{12}+\lambda_2\lambda_3q_{23}+\lambda_1\lambda_3q_{13}>0.
\]
We prove the three assertions in turn.

\medskip
\noindent\textit{Mirkin distance.}
Since $A=\{V\}$, $E_0(A)=\varnothing$. Thus
\[
\dMK(A,B)=|E_0(B)|
\]
for every partition $B$. Under $B_1=\{g_1,\ g_2\cup g_3\}$, the cut edges are exactly those between $g_1$ and $g_2\cup g_3$, so
\[
\dMK(A,B_1)=\mathsf C_{12}+\mathsf C_{13}.
\]
Under $B_2=\{g_3,\ g_1\cup g_2\}$, the cut edges are exactly those between $g_3$ and $g_1\cup g_2$, so
\[
\dMK(A,B_2)=\mathsf C_{13}+\mathsf C_{23}.
\]
Therefore
\[
\dMK(A,B_1)-\dMK(A,B_2)=\mathsf C_{12}-\mathsf C_{23}.
\]
Dividing by $N^2$ and passing to the almost sure limit gives
\[
N^{-2}\{\dMK(A,B_1)-\dMK(A,B_2)\}
\to
\lambda_1\lambda_2q_{12}-\lambda_2\lambda_3q_{23}
=
\delta_{\kappa}^{\mathrm{inter}}.
\]
This limit is strictly positive when $\lambda_1q_{12}>\lambda_3q_{23}$.

\medskip
\noindent\textit{Graph-informed variation of information.}
Under the marked construction, $\map(A)=\{\{\star\},E\}$. For $B_1$, write
\[
\mathsf a_{1,N}=\mathsf W_1,\qquad
\mathsf b_{1,N}=\mathsf W_2+\mathsf W_3+\mathsf C_{23},\qquad
\mathsf c_{1,N}=\mathsf C_{12}+\mathsf C_{13}.
\]
Thus the three blocks of $\map(B_1)$ have sizes $\mathsf a_{1,N}$, $\mathsf b_{1,N}$ and $\mathsf c_{1,N}+1$, where the $+1$ is due to the marker $\star$ in the cut block. For $B_2$, write
\[
\mathsf a_{2,N}=\mathsf W_3,\qquad
\mathsf b_{2,N}=\mathsf W_1+\mathsf W_2+\mathsf C_{12},\qquad
\mathsf c_{2,N}=\mathsf C_{23}+\mathsf C_{13}.
\]
The three blocks of $\map(B_2)$ have sizes $\mathsf a_{2,N}$, $\mathsf b_{2,N}$ and $\mathsf c_{2,N}+1$.

Let $h_2(x)=-x\log x-(1-x)\log(1-x)$. Since variation of information may be written as
$
\HH\{\map(B_k)\mid \map(A)\}+\HH\{\map(A)\mid \map(B_k)\},
$
we have, for $k=1,2$,
\[
\dVI(A,B_k)
=
\frac{\mathsf M}{\mathsf M+1}
\HH\left(
\frac{\mathsf a_{k,N}}{\mathsf M},
\frac{\mathsf b_{k,N}}{\mathsf M},
\frac{\mathsf c_{k,N}}{\mathsf M}
\right)
+
\frac{\mathsf c_{k,N}+1}{\mathsf M+1}
h_2\left(\frac{1}{\mathsf c_{k,N}+1}\right).
\]
Indeed, conditional on the real-edge block $E$ of $\map(A)$, the partition $\map(B_k)$ has proportions $\mathsf a_{k,N}/\mathsf M$, $\mathsf b_{k,N}/\mathsf M$ and $\mathsf c_{k,N}/\mathsf M$, while the only block of $\map(B_k)$ not contained in a block of $\map(A)$ is the marked cut block of size $\mathsf c_{k,N}+1$, which contains $\star$ and $\mathsf c_{k,N}$ real edges.

By the almost sure limits above,
\[
N^{-2}(\mathsf a_{1,N},\mathsf b_{1,N},\mathsf c_{1,N})\to (W_1,U_{23},C_{1\cdot})\quad \text{a.s.},
\]
and
\[
N^{-2}(\mathsf a_{2,N},\mathsf b_{2,N},\mathsf c_{2,N})\to (W_3,U_{12},C_{\cdot3})\quad \text{a.s.}.
\]
Moreover, for $k=1,2$,
\[
0\leq
\frac{\mathsf c_{k,N}+1}{\mathsf M+1}
h_2\left(\frac{1}{\mathsf c_{k,N}+1}\right)
\leq
\frac{\log(\mathsf c_{k,N}+1)+1}{\mathsf M+1}
\to 0
\]
almost surely. Hence, by continuity of entropy on the simplex,
\[
\dVI(A,B_1)\to \HH(W_1/T,U_{23}/T,C_{1\cdot}/T)\quad \text{a.s.},
\]
and
\[
\dVI(A,B_2)\to \HH(W_3/T,U_{12}/T,C_{\cdot3}/T)\quad \text{a.s.}.
\]
Therefore
\[
\dVI(A,B_1)-\dVI(A,B_2)
\to
\delta_{\mathrm{VI}}^{\mathrm{inter}}
\quad \text{a.s.}
\]

It remains to show that $\delta_{\mathrm{VI}}^{\mathrm{inter}}>0$ under the stated assumptions. Set
\[
a_1=W_1,\quad b_1=U_{23},\quad c_1=C_{1\cdot},\qquad
a_2=W_3,\quad b_2=U_{12},\quad c_2=C_{\cdot3}.
\]
The assumptions imply $b_1>\max\{a_1,c_1\}$ and $b_2>\max\{a_2,c_2\}$. They also imply
\[
a_1-a_2=W_1-W_3=(\lambda_1^2p_1-\lambda_3^2p_3)/2\geq0
\]
and
\[
c_1-c_2=\lambda_1\lambda_2q_{12}-\lambda_2\lambda_3q_{23}
=
\lambda_2(\lambda_1q_{12}-\lambda_3q_{23})>0.
\]
Let $\Delta_a=a_1-a_2\geq0$ and $\Delta_c=c_1-c_2>0$. Since $a_1+b_1+c_1=a_2+b_2+c_2=T$,
\[
a_2=a_1-\Delta_a,\qquad
b_2=b_1+\Delta_a+\Delta_c,\qquad
c_2=c_1-\Delta_c.
\]
Let $f(x)=x\log x$ on $(0,\infty)$. The function $f$ is strictly convex, and if $v\geq u>0$ and $0<\delta\leq u$, then
\[
f(u-\delta)+f(v+\delta)>f(u)+f(v).
\]
This follows since $g(\delta)=f(u-\delta)+f(v+\delta)$ satisfies
$
g'(\delta)=f'(v+\delta)-f'(u-\delta)=\log(v+\delta)-\log(u-\delta)>0.
$
Applying this first with $u=a_1$, $v=b_1$, and $\delta=\Delta_a$ gives, with equality allowed only if $\Delta_a=0$,
\[
f(a_2)+f(b_1+\Delta_a)\geq f(a_1)+f(b_1).
\]
Applying it next with $u=c_1$, $v=b_1+\Delta_a$, and $\delta=\Delta_c$ gives the strict inequality
\[
f(c_2)+f(b_2)>f(c_1)+f(b_1+\Delta_a),
\]
because $\Delta_c>0$ and $b_1+\Delta_a\geq b_1>c_1$. Adding yields
\[
f(a_2)+f(b_2)+f(c_2)>f(a_1)+f(b_1)+f(c_1).
\]
Since, for positive $a,b,c$ with $a+b+c=T$,
\[
\HH(a/T,b/T,c/T)=-\frac{f(a)+f(b)+f(c)}{T}+\log T,
\]
we obtain
\[
\begin{aligned}
\delta_{\mathrm{VI}}^{\mathrm{inter}}
={}&\HH(a_1/T,b_1/T,c_1/T)-\HH(a_2/T,b_2/T,c_2/T)\\
={}&\frac{f(a_2)+f(b_2)+f(c_2)-f(a_1)-f(b_1)-f(c_1)}{T}>0.
\end{aligned}
\]

\medskip
\noindent\textit{Graph-informed van Dongen distance.}
Again $\map(A)=\{\{\star\},E\}$. Let $Q$ be a marked edge partition with real-edge class sizes $a,b,c$, where the marked cut block has size $c+1$ and the other two blocks have sizes $a$ and $b$. For all sufficiently large $N$ in the present setting, $c\geq1$. Then the van Dongen formula gives
\[
\phi_{\mathrm{VD}}\{\map(A),Q\}
=
2-\frac{1}{\mathsf M+1}
\left[
1+\max\{a,b,c\}+a+b+c
\right]
=
1-\frac{\max\{a,b,c\}}{\mathsf M+1}.
\]
For $B_1$, the real-edge class sizes are
\[
\mathsf W_1,\qquad
\mathsf W_2+\mathsf W_3+\mathsf C_{23},\qquad
\mathsf C_{12}+\mathsf C_{13}.
\]
For $B_2$, the real-edge class sizes are
\[
\mathsf W_3,\qquad
\mathsf W_1+\mathsf W_2+\mathsf C_{12},\qquad
\mathsf C_{23}+\mathsf C_{13}.
\]
The strict inequalities $U_{23}>\max\{C_{1\cdot},W_1\}$ and $U_{12}>\max\{C_{\cdot3},W_3\}$ imply, by the almost sure convergence of normalized edge counts, that almost surely for all sufficiently large $N$,
\[
\max\{\mathsf W_1,\mathsf W_2+\mathsf W_3+\mathsf C_{23},\mathsf C_{12}+\mathsf C_{13}\}
=
\mathsf W_2+\mathsf W_3+\mathsf C_{23},
\]
and
\[
\max\{\mathsf W_3,\mathsf W_1+\mathsf W_2+\mathsf C_{12},\mathsf C_{23}+\mathsf C_{13}\}
=
\mathsf W_1+\mathsf W_2+\mathsf C_{12}.
\]
Thus, almost surely for all sufficiently large $N$,
\[
\dVD(A,B_1)=1-\frac{\mathsf W_2+\mathsf W_3+\mathsf C_{23}}{\mathsf M+1},
\qquad
\dVD(A,B_2)=1-\frac{\mathsf W_1+\mathsf W_2+\mathsf C_{12}}{\mathsf M+1}.
\]
Subtracting and taking limits gives
\[
\dVD(A,B_1)-\dVD(A,B_2)
\to
\frac{U_{12}-U_{23}}{T}
=
\delta_{\mathrm{VD}}^{\mathrm{inter}}
\quad \text{a.s.}
\]
Finally,
\[
U_{12}-U_{23}=W_1+\lambda_1\lambda_2q_{12}-W_3-\lambda_2\lambda_3q_{23},
\]
so the assumed inequality $W_1+\lambda_1\lambda_2q_{12}>W_3+\lambda_2\lambda_3q_{23}$ implies $\delta_{\mathrm{VD}}^{\mathrm{inter}}>0$. This proves the theorem.
\end{proof}

\section{Proof of Theorem \ref{thm:intra}}
\begin{proof}
Let $L=g_1\cup g_2$ and $R=g_3\cup g_4$. Let $\mathsf W_{11}$ and $\mathsf W_{22}$ denote the numbers of observed edges inside $g_1$ and $g_2$, respectively, let $\mathsf W_{33}$ and $\mathsf W_{44}$ denote the numbers of observed edges inside $g_3$ and $g_4$, respectively, let $\mathsf C_{12}$ and $\mathsf C_{34}$ denote the numbers of observed edges between $g_1,g_2$ and between $g_3,g_4$, respectively, and let $\mathsf C_{LR}$ denote the number of observed edges with one endpoint in $L$ and the other in $R$. Write
\[
\mathsf W_L=\mathsf W_{11}+\mathsf W_{22}+\mathsf C_{12},\qquad
\mathsf W_R=\mathsf W_{33}+\mathsf W_{44}+\mathsf C_{34},
\]
and
\[
\mathsf M=|E|=\mathsf W_L+\mathsf W_R+\mathsf C_{LR}.
\]

We normalize edge counts by $n_1n_2$. Since $r=n_1/n_2$ is fixed and $n_1,n_2\to\infty$, the strong law of large numbers gives
\[
\frac{\mathsf W_{11}}{n_1n_2}\to \frac{p_1r}{2},\qquad
\frac{\mathsf W_{22}}{n_1n_2}\to \frac{p_1}{2r},\qquad
\frac{\mathsf C_{12}}{n_1n_2}\to p_1\qquad \text{a.s.},
\]
and
\[
\frac{\mathsf W_{33}}{n_1n_2}\to \frac{p_2r}{2},\qquad
\frac{\mathsf W_{44}}{n_1n_2}\to \frac{p_2}{2r},\qquad
\frac{\mathsf C_{34}}{n_1n_2}\to p_2\qquad \text{a.s.}.
\]
Moreover,
\[
\frac{\mathsf W_L}{n_1n_2}\to p_1S_1(r),\qquad
\frac{\mathsf W_R}{n_1n_2}\to p_2S_1(r)\qquad \text{a.s.},
\]
where $S_1(r)=1+\{r+r^{-1}\}/2$, and
\[
\frac{\mathsf C_{LR}}{n_1n_2}\to qK(r)\qquad \text{a.s.},
\]
where $K(r)=r+r^{-1}+2$. Therefore, on an event of probability one that we fix throughout the proof,
\[
\frac{\mathsf M}{n_1n_2}\to D(r)=(p_1+p_2)S_1(r)+qK(r)>0.
\]

We now prove the three claims.

\medskip
\noindent\textit{Mirkin distance.}
The partitions $B_1$ and $B_2$ are refinements of $A$. Hence the Mirkin distance from $A$ to a refinement counts exactly the newly created cut edges. Under $B_1$, the only newly cut edges are those between $g_1$ and $g_2$, while under $B_2$ the only newly cut edges are those between $g_3$ and $g_4$. Thus
\[
\dMK(A,B_1)=\mathsf C_{12},\qquad \dMK(A,B_2)=\mathsf C_{34}.
\]
Therefore
\[
\dMK(A,B_1)-\dMK(A,B_2)=\mathsf C_{12}-\mathsf C_{34},
\]
and
\[
\frac{\dMK(A,B_1)-\dMK(A,B_2)}{n_1n_2}\to p_1-p_2\qquad \text{a.s.}
\]
Since $n_1n_2/N^2=r\{4(r+1)^2\}^{-1}$, this is equivalent to
\[
N^{-2}\{\dMK(A,B_1)-\dMK(A,B_2)\}\to \frac{r}{4(r+1)^2}(p_1-p_2)\qquad \text{a.s.},
\]
which is strictly positive when $p_1>p_2$.

\medskip
\noindent\textit{Graph-informed variation of information.}
Let $P_A=\map(A)$, $P_1=\map(B_1)$ and $P_2=\map(B_2)$ denote the marked edge partitions of $E^\star$. Under $A$, the marked edge partition has three blocks: the edges internal to $L$, the edges internal to $R$, and the marked cut block consisting of the edges across $L$ and $R$ together with $\star$.

We compute $\dVI(A,B_1)=\HH(P_1\mid P_A)+\HH(P_A\mid P_1)$. Under $B_1$, the block internal to $L$ is split into the edges inside $g_1$, the edges inside $g_2$, and the edges between $g_1$ and $g_2$; the last group is absorbed into the marked cut block of $P_1$ together with the original edges across $L$ and $R$ and the marker $\star$. Hence, for all sufficiently large $n_1,n_2$ on the fixed almost sure event,
\[
\HH(P_1\mid P_A)
=
\frac{\mathsf W_L}{\mathsf M+1}
\HH\left(
\frac{\mathsf W_{11}}{\mathsf W_L},
\frac{\mathsf W_{22}}{\mathsf W_L},
\frac{\mathsf C_{12}}{\mathsf W_L}
\right),
\]
and
\[
\HH(P_A\mid P_1)
=
\frac{\mathsf C_{12}+\mathsf C_{LR}+1}{\mathsf M+1}
h_2\left(
\frac{\mathsf C_{12}}{\mathsf C_{12}+\mathsf C_{LR}+1}
\right),
\]
where $h_2(x)=-x\log x-(1-x)\log(1-x)$. The $+1$ terms come from the marker in the cut block.

By the almost sure limits above,
\[
\frac{\mathsf W_L}{\mathsf M+1}
\HH\left(
\frac{\mathsf W_{11}}{\mathsf W_L},
\frac{\mathsf W_{22}}{\mathsf W_L},
\frac{\mathsf C_{12}}{\mathsf W_L}
\right)
\to
\frac{p_1S_1(r)}{D(r)}h_3(r)\qquad \text{a.s.},
\]
where
\[
h_3(r)=-u_0(r)\log u_0(r)-u_1(r)\log u_1(r)-u_2(r)\log u_2(r),
\]
with $u_0(r)=S_1(r)^{-1}$, $u_1(r)=r\{2S_1(r)\}^{-1}$ and $u_2(r)=\{2rS_1(r)\}^{-1}$. Also,
\[
\frac{\mathsf C_{12}+\mathsf C_{LR}+1}{\mathsf M+1}
h_2\left(
\frac{\mathsf C_{12}}{\mathsf C_{12}+\mathsf C_{LR}+1}
\right)
\to
\frac{p_1+qK(r)}{D(r)}
h_2\left(
\frac{p_1}{p_1+qK(r)}
\right)
\qquad \text{a.s.}
\]
Therefore
\[
\dVI(A,B_1)\to
\frac{\{p_1+qK(r)\}h_2[p_1/\{p_1+qK(r)\}]+p_1S_1(r)h_3(r)}{D(r)}
=
\frac{F_r(p_1)}{D(r)}
\qquad \text{a.s.}
\]

The same argument with the roles of $L$ and $R$ exchanged gives
\[
\dVI(A,B_2)\to
\frac{\{p_2+qK(r)\}h_2[p_2/\{p_2+qK(r)\}]+p_2S_1(r)h_3(r)}{D(r)}
=
\frac{F_r(p_2)}{D(r)}
\qquad \text{a.s.}
\]
Hence
\[
\dVI(A,B_1)-\dVI(A,B_2)
\to
\frac{F_r(p_1)-F_r(p_2)}{D(r)}
\qquad \text{a.s.}
\]

It remains to show that $F_r$ is strictly increasing. Write $c=qK(r)>0$ and $s=S_1(r)h_3(r)>0$, so that
\[
F_r(p)=(p+c)h_2\left(\frac{p}{p+c}\right)+sp.
\]
A direct derivative calculation gives
\[
\frac{d}{dp}\left[(p+c)h_2\left(\frac{p}{p+c}\right)\right]
=
\log\left(\frac{p+c}{p}\right),
\]
and therefore
\[
F_r'(p)=\log\left(\frac{p+c}{p}\right)+s>0
\]
for every $p>0$. Thus $F_r$ is strictly increasing, so $p_1>p_2$ implies $F_r(p_1)>F_r(p_2)$. Since $D(r)>0$, the limiting difference is strictly positive.

\medskip
\noindent\textit{Graph-informed van Dongen distance.}
We compute the van Dongen limits from the marked edge partitions. Since the ground set is $E^\star$, the denominator in the van Dongen distance is $\mathsf M+1$. Consider first $P_A=\map(A)$ and $P_1=\map(B_1)$. The partition $P_A$ has three blocks: the edges internal to $L$, the edges internal to $R$, and the marked cut block consisting of the edges across $L$ and $R$ together with $\star$. The partition $P_1$ has four blocks: the edges inside $g_1$, the edges inside $g_2$, the edges inside $R$, and the marked cut block consisting of the edges between $g_1$ and $g_2$, the edges across $L$ and $R$, and $\star$.

For the first max-sum, the block of $P_A$ internal to $L$ intersects the blocks of $P_1$ in three pieces with sizes $\mathsf W_{11}$, $\mathsf W_{22}$ and $\mathsf C_{12}$; hence its maximal intersection has limiting mass $p_1M_1(r)$, where
\[
M_1(r)=\max\left\{\frac r2,\ \frac{1}{2r},\ 1\right\}.
\]
The block of $P_A$ internal to $R$ intersects only the corresponding within-$R$ block of $P_1$, giving limiting mass $p_2S_1(r)$. The marked cut block of $P_A$ intersects only the marked cut block of $P_1$, with size $\mathsf C_{LR}+1$, giving limiting mass $qK(r)$. Therefore
\[
\frac{1}{n_1n_2}\sum_i\max_j |(P_A)_i\cap (P_1)_j|
\to
p_1M_1(r)+p_2S_1(r)+qK(r)
\qquad \text{a.s.}
\]

For the second max-sum, the within-$g_1$ block of $P_1$ contributes $\mathsf W_{11}$, the within-$g_2$ block contributes $\mathsf W_{22}$, the within-$R$ block contributes $\mathsf W_R$, and the marked cut block of $P_1$ contributes $\max\{\mathsf C_{12},\mathsf C_{LR}+1\}$. Hence
\[
\frac{1}{n_1n_2}\sum_j\max_i |(P_A)_i\cap (P_1)_j|
\to
\frac{p_1r}{2}+\frac{p_1}{2r}+p_2S_1(r)+\max\{p_1,qK(r)\}
\qquad \text{a.s.}
\]
Using $(\mathsf M+1)/(n_1n_2)\to D(r)$, we obtain
\[
\dVD(A,B_1)\to
\frac{p_1\Phi(r)+qK(r)-\max\{p_1,qK(r)\}}{D(r)}
\qquad \text{a.s.},
\]
where
\[
\Phi(r)=2+\frac{r+r^{-1}}{2}-M_1(r).
\]

The same calculation with $P_2=\map(B_2)$ gives
\[
\dVD(A,B_2)\to
\frac{p_2\Phi(r)+qK(r)-\max\{p_2,qK(r)\}}{D(r)}
\qquad \text{a.s.}
\]
Therefore
\[
\dVD(A,B_1)-\dVD(A,B_2)
\to
\frac{(p_1-p_2)\Phi(r)+\max\{p_2,qK(r)\}-\max\{p_1,qK(r)\}}{D(r)}
\qquad \text{a.s.}
\]

Under the additional condition $p_2>qK(r)$, and since $p_1>p_2$, we have
\[
\max\{p_2,qK(r)\}=p_2,\qquad \max\{p_1,qK(r)\}=p_1.
\]
Hence
\[
\dVD(A,B_1)-\dVD(A,B_2)
\to
\frac{(p_1-p_2)\{\Phi(r)-1\}}{D(r)}
\qquad \text{a.s.}
\]
Finally, $\Phi(r)>1$ for every $r>0$: if $1/2\leq r\leq2$, then $M_1(r)=1$ and $\Phi(r)=1+(r+r^{-1})/2>1$; if $r>2$, then $M_1(r)=r/2$ and $\Phi(r)=2+r^{-1}/2>1$; if $0<r<1/2$, then $M_1(r)=1/(2r)$ and $\Phi(r)=2+r/2>1$. Thus the limit is strictly positive. This completes the proof.
\end{proof}

\end{document}